\documentclass[runningheads,a4paper]{llncs}

\usepackage{amssymb}
\usepackage{graphicx}

\def\ds{\displaystyle}

\def\a{\alpha}
\def\b{\beta}

\def\s{\sigma}
\def\t{\tau}
\def\tp{\tau_p}
\def\r{\rho}

\def\E{{\bf E}}
\def\Ex{\E^x}

\def\F{{\cal F}}
\def\R{{\mathbb R}}
\def\L{{\mathbb L}}

\def\M{{\mathcal M}}
\def\I{{\mathcal I}}
\def\II{\mbox{\bf 1}}

\def\beq{\begin{equation}}
\def\eeq{\end{equation}}

\usepackage{url}
\urldef{\mailsa}\path|{alfred.hofmann, ursula.barth, ingrid.haas, frank.holzwarth,|
\urldef{\mailsb}\path|anna.kramer, leonie.kunz, christine.reiss, nicole.sator,|
\urldef{\mailsc}\path|erika.siebert-cole, peter.strasser, lncs}@springer.com|
\newcommand{\keywords}[1]{\par\addvspace\baselineskip
\noindent\keywordname\enspace\ignorespaces#1}

\begin{document}

\mainmatter  

\title{Real Options and Threshold Strategies}

\titlerunning{Real Options and Threshold Strategies}

%
%
\author{Vadim Arkin%
\and Alexander Slastnikov}
\authorrunning{Real Options and Threshold Strategies}

\institute{Central Economics and Mathematics Institute, Russian Academy of Sciences,\\
Nakhimovskii pr. 47, 117418 Moscow, Russia\\
}

%
%

\toctitle{Lecture Notes in Computer Science}
\tocauthor{Authors' Instructions}
\maketitle

\begin{abstract}
The paper considers an investment timing problem appearing in real options theory. Present values from an investment project are modeled by general diffusion process. We prove necessary and sufficient conditions under which an optimal investment time is induced by threshold strategy. We study also the conditions of optimality of threshold strategy (over all threshold strategies) and discuss the connection between solutions to  investment timing problem and to free-boundary problem.
\end{abstract}

\section{Introduction}

One of the fundamental problems in real options theory concerns the determination of optimal time for investment into a given project (see, e.g., classical monograph \cite{DP}).

Let us consider an investment project, for example, a creating of new firm in the real sector of economy. This project is characterized by a
pair $(X_t ,t \ge 0,\,\,I)$, where $X_t$ is a present value of the
firm created at time $t$, and $I$ is a cost of investment required to implement the project (for example, to create the firm). Prices on input and output production are assumed to be stochastic, so  $X_t$ is considered as a stochastic process,
defined at a probability space with filtration  $(\Omega, \F,
\{\F_t, t \ge 0\},{\bf P})$. This model supposes that:

\begin{itemize}
\itemsep=0pt
\item[-] at any moment,  a decision-maker (investor)  can either
{\it accept} the project and proceed  with  the investment  or
{\it delay} the decision until he obtains  new information;
\item[-]
investment are considered to be instantaneous  and irreversible so
that  they cannot be withdrawn from the project any more and used
for other purposes.
\end{itemize}

The investor's problem is to evaluate the project and to determine an
appropriate time for the investment (investment timing problem). In real option theory investment times are considered as stopping times (regarding to flow of $\s$-algebras $\{\F_t, t \ge 0\}$).

In real options theory there are two different approaches to solving investment timing problem (see \cite{DP}).

The value of project under the first approach is the maximum of
net present value (NPV) from the implemented project over all stopping times (investment rules):
\beq\label{1}
F = \mathop {\max }\limits_\tau{\E}(X_\tau -
I)e^{-\rho \tau },
\eeq
where  $\r$ is the given discount rate. An optimal stopping
time $\t^*$ in (\ref{1}) is viewed as optimal investment time
(investment rule).

Within the second approach an opportunity to invest is considered
as an American call option -- the right but not obligation to buy the asset on
predetermined price. At that an exercise time is viewed as
investment time, and value of option is accepted as a value of
investment project. In these framework a project is spanned with
some traded asset, which price is completely correlated
with present value of the project $X_t$.
In order to evaluate a (rational) value of this real option one can use methods of financial options pricing theory, especially, contingent claims analysis (see, e.g., \cite{DP}).

In this paper we follow the approach that optimal investment timing decision can be mathematically determined as a solution of an optimal stopping problem (\ref{1}). Such an approach started from the well-known McDonald--Siegel model (see \cite{MS}, \cite{DP}), in which the underlying present value's dynamics is modeled by a geometric Brownian motion. The majority of results on this problem (optimal investment strategy) has a threshold structure: to invest when present value from the project exceeds the certain level (threshold).
In the heuristic level this is so for the cases of geometric Brownian motion, arithmetic Brownian motion, mean-reverting process and some other (see \cite{DP}). And the general question arises: For what underlying processes  an optimal decision to an investment timing problem will have a threshold structure?

Some sufficient conditions in this direction was obtained in \cite{Al}. In this paper we focus on necessary and sufficient conditions for optimality of threshold strategies in investment timing problem. Since this problem is a special case of optimal stopping problem, the similar question may be addressed to a general optimal stopping problem: Under what conditions (on both process and payoff function) an optimal stopping time will have a threshold structure? Some results in this direction (in the form of necessary and sufficient conditions) were obtained in \cite{AS}, \cite{CM}, \cite{A} under some additional assumptions on underlying process and/or payoffs.

The paper is organized as follows. After a formal description of investment timing problem and assumptions on underlying process (Section 2.1) we go to study of threshold strategies for this problem. Since an investment timing problem in threshold strategies is reduced to one-dimensional maximization, a related problem is to find an optimal threshold. In Section 2.2 we give necessary and sufficient conditions for optimal threshold (over all thresholds).
Solving a  free-boundary problem (based on smooth-pasting principle)  is the most accepted method (but not the only method, see, e.g., \cite{PS}) that allows to find a solution to optimal stopping problem. In Section 2.3 we discuss the connection between solutions to investment timing problem and to free-boundary problem. Finally, in Section 2.4 we prove the main result on necessary and sufficient conditions under which an optimal investment time is generated by threshold strategy.

\section{Investment Timing Problem}

Let $I$ be a cost of investment required for implementing a project, and $X_t$ is the present value from the project started at time $t$. As usual investment supposed to be instantaneous and irreversible, and the project --- infinitely-lived.

At any time a decision-maker (investor)  can either {\it accept} the project and proceed  with  the investment  or
{\it delay} the decision until he/she obtains  new information
regarding its environment (prices of the product and resources,
demand etc.). The goal of a decision-maker in this
situation is to find, using the available information, an
optimal time for investing the project (investment timing
problem), which maximizes the net present value from the project:
\beq
\label{optinv0}
{\Ex}\left(X_\t - I\right) e^{-\r \t}\II_{\{\t{<}\infty\}} \to \max_{\t\in\M},
\eeq
where $\Ex$ means the expectation for the process $X_t$ starting from the
initial state $x$, $\II_A$ is indicator function of the set $A$,
and the maximum is considered over stopping times $\t$ from a certain class $\M$ of stopping times\footnote{In this paper we consider stopping times which can take infinite values (with positive probability)}.

We consider the case $I<r$ else the optimal time in (\ref{optinv0}) will be $+\infty$.

\subsection{Mathematical Assumptions}

Let $X_t$ be a diffusion process with values in the  interval  $D\subseteq \R^1$ with boundary points $l$ and $r$, where $ -\infty \le l< r\le +\infty$, open or closed (i.e. it may be $(l,r)$,\, $[l,r)$,\, $(l,r]$,\, or $[l,r]$),  which is a solution to stochastic differential equation:
\beq \label{infop1}
dX_t=a(X_t)dt + \s(X_t) dw_t,\quad X_0=x,
\eeq
where $w_t$ is a standard Wiener process, $a:\ D \to \R^1$ and  $\s:\ D \to \R^1_+$ are the drift and diffusion functions, respectively. Denote $\I={\rm int} (D)=(l,r)$.

The process $X_t$ is assumed to be regular; this means
that, starting from an arbitrary point $x\in \I$, this process
reaches any point $y\in \I$ in finite time with positive
probability.

It is known that the following local integrability condition:
\beq \label{regular}
 \int_{x-\varepsilon}^{x+\varepsilon}\frac{1+|a(y)|}{\s^2(y)} dy <\infty \quad \mbox{for some }\varepsilon>0,
\eeq
at any $x\in \I$  guarantees the existence of weak solution of equation (\ref{infop1}) and its regularity (see, e.g. \cite{KS}).
\smallskip

The process $X_t$ is associated with infinitesimal operator
\beq\label{operator}
\L f(x)= a(x)f'(x) +\frac12\s^2(x) f''(x).
\eeq

Under the  condition (\ref{regular})
there exist (unique up to constant positive multipliers) increasing and decreasing functions $\psi(x)$ and $\varphi(x)$ with absolutely continuous derivatives, which are the fundamental  solutions to the ODE
\beq\label{diffur}
\L f(p)=\r f(p)
\eeq
almost sure (in Lebesque measure) on the interval $\I$ (see, e.g. \cite[Chapter 5, Lemma 5.26]{KS}). Moreover, $0<\psi(p),\, \varphi(p)<\infty$ for $p\in \I$.
Note, if functions $a(x),\ \s(x)$ are continuous, then \ $\psi,\, \varphi\in C^2(\I)$.

\subsection{Optimality of Threshold Strategies}

Let us define $\tp=\tp(x)=\inf\{t{\ge} 0: X_t\ge p\}$ --- the first time when the process $X_t$, starting from $x$, \emph{exceeds} level  $p$.  We will call $\tp$ as threshold stopping time generated by the threshold strategy  ---  to stop when the process exceeds threshold $p$. Let $\M_{\rm th}=\{\tp, \ p\in \I\}$ be a class of all such threshold stopping times.

For the class  $\M_{\rm th}$ of threshold stopping times the investment timing problem (\ref{optinv0}) can be written as follows:
\beq
\label{optinv1}
\left(p - I\right) {\Ex}e^{-\r \tp} \to \max_{p\in(l,r)}.
\eeq

Such a problem appeared in \cite{DPS} as the heuristic method for solving a general investment timing problem (\ref{optinv0}) over class of all stopping times.

We say that threshold  $p^*$ is optimal for the investment timing problem  (\ref{optinv1}) if threshold stopping time $\t_{p^*}$ is optimal in (\ref{optinv1}).
The following result gives necessary and sufficient conditions for optimal threshold.

\begin{theorem}
Threshold $p^*\in\I$ is optimal in the problem \emph{(\ref{optinv1})} for all $x\in \I$, if and only if the following conditions hold:
\begin{eqnarray}
&& \frac{p-I}{\psi(p)}\le \frac{p^*-I}{\psi(p^*)}\quad  \mbox{\rm whenever }  p<p^*; \label{criteria0}\\
&& \frac{p-I}{\psi(p)}\quad \mbox{\rm does not increase for }  p\ge p^* \label{criteria00},
\end{eqnarray}
where $\psi(p)$ is an increasing solution to ODE \emph{(\ref{diffur})}.
\end{theorem}

\begin{proof} Let us denote the left-hand side in (\ref{optinv1}) as $V(p;x)$.
Obviously, $V(p;x)=x-I$ for $x\ge p$.

Along with the above stopping time let us define the first hitting time to threshold:  $T_p= \inf\{t{\ge} 0: X_t= p\},\ p\in (l,r)$.

For $x<p$, obviously, $\tp=T_p$ and using known formula $\ds \Ex e^{-\r T_p}= \psi(x)/\psi(p)$ (see, e.g., \cite{IM}, \cite{BS})we have:
\beq\label{Vpx}
V(p;x)=(p-I)\Ex e^{-\r\tp} \II_{\{\tp{<}\infty\}}=(p-I)\Ex e^{-\r T_p} =\frac{p-I}{\psi(p)}\psi(x).
\eeq

Denote $h(x)=(x-I)/\psi(x)$.

i) Let $p^*\in\I$ be an optimal  threshold in the problem (\ref{optinv1}) for all $x\in \I$. Then for $p<p^*$ we have
$$
V(p;p)=p-I\le V(p^*;p)=\frac{p^*-I}{\psi(p^*)} \psi(p),
$$
i.e. (\ref{criteria0}) holds. If $p^*\le p_1 < p_2$, then
$$
V(p_2;p_1)=h(p_2) \psi(p_1)\le V(p^*;p_1)= p_1-I=
h(p_1) \psi(p_1),
$$
and it follows (\ref{criteria00}).

ii) Now, let (\ref{criteria0})--(\ref{criteria00}) hold.

Let $p<p^*$.
If $x \ge p^*$, then $V(p;x)= x-I=
V(p^*;x)$.

If $p\le x<{p^*}$, then, due to (\ref{criteria0}), $V_p(x)=x-I=h(x)\psi(x) \le h(p^*)\psi(x)= V(p^*;x)$.

Finally, if $x<p$, then, using (\ref{criteria0}) and (\ref{Vpx}), we have: $V(p;x)= h(p)\psi(x)\le h(p^*)\psi(x)= V(p^*;x)$.

Consider the case $p>p^*$. If $x\ge p$, then  $V(p;x)= x-I=
V(p^*;x)$.

Whenever $p^*\le x < p$, then, due to (\ref{criteria00}),
$V(p;x)= h(p)\psi(x)\le h(x)\psi(x)= x-I=V(p^*;x)$.

When $x < p^*$, then $V(p;x)= h(p)\psi(x)\le h(p^*)\psi(x)= V(p^*;x)$, since $h(p)\le h(p^*)$.

Theorem is completely proved.
\end{proof}

\begin{remark} The condition (\ref{criteria00}) is equivalent to the inequality
$$
(p -I) \psi'(p) \ge  \psi(p) \quad {\rm for } \ p\ge p^*.
$$
This relation implies, in particular, that optimal threshold $p^*$ must be strictly greater than investment cost $I$ (because $\psi(p^*),\ \psi'(p^*)$ are positive values).
\end{remark}

\begin{remark} Assume that $\log \psi(x)$ is a convex function, i.e. $\psi'(x)/\psi(x)$ increases. For this case there exists a unique point $p^*$ which satisfies the equation
\beq\label{smooth}
  (p^* -I) \psi'(p^*) = \psi(p^*)
\eeq
and constitutes the optimal threshold  in the problem (\ref{optinv1}) for all $x\in \I$. Indeed, the sign of derivative of the function $(p-I)/\psi(p)$ coincides with the sign of $\psi(p)-(p -I) \psi'(p)$. Therefore, in the considered case the conditions (\ref{criteria0})--(\ref{criteria00}) in Theorem 1 are true automatically.
\end{remark}


We can give a number of cases of diffusion processes which are more or less realistic for modeling a process $X_t$ of present values from a project. Some of them are presented below.
\smallskip

1) \emph{Geometric Brownian motion (GBM)}:
\beq \label{GBM}
dX_t=X_t(\a dt + \s dw_t).
\eeq

For this case $\ds \psi(x)=x^\b$, where  $\b$ is the positive root of the equation\\ $\ds \textstyle\frac12\s^2\b(\b-1)+\a\b-\r = 0$.

2) \emph{Arithmetic Brownian motion (ABM)}:
\beq \label{ABM}
dX_t=x+ \a dt + \s \,dw_t.
\eeq

For this case $\ds \psi(x)=e^{\b x}$, where  $\b$ is the positive root of the equation\\ $\ds \textstyle\frac12\s^2\b^2+\a\b-\r = 0$.

3) \emph{Mean-reverting process \ (or geometric Ornstein--Uhlenbeck process)}:
\beq \label{MRP}
dX_t=\a (\bar x-X_t)X_tdt + \s X_t\, dw_t.
\eeq

For this case $\ds \psi(x)=x^\b {}_1\! F_1\left(\b,2\b+\frac{2\a \bar x}{\s^2}; \frac{2\a}{\s^2}x\right) $, where  $\b$ is the positive root of  equation $\ds \textstyle\frac12\s^2\b(\b-1)+\a\bar x\b-\r = 0$, and ${}_1\! F_1(p,q;x)$ is confluent hypergeometric function satisfying Kummer's equation $\ds xf''(x)+(q-x)f'(x)-pf(x)=0$.

4) \emph{Square-root mean-reverting process \ (or Cox--Ingersoll--Ross process)}:
\beq \label{SRMRP}
dX_t=\a (\bar x-X_t)dt + \s \sqrt {X_t}\, dw_t.
\eeq

For this case $\ds \psi(x)= {} _1\!F_1\left(\frac{\r}{\a},\frac{2\a \bar x}{\s^2}; \frac{2\a}{\s^2}x\right) $.
\smallskip

The above processes are well studied in the literature (in connection with real options and optimal stopping problems see, for example, \cite{DP}, \cite{JZ}).

For the first two processes (\ref{GBM}) and (\ref{ABM}) the conditions of Theorem 1 give explicit formulas for optimal threshold in investment timing problem:
$$
p^*=\frac{\b}{\b-1}I \ \  {\rm for  \ GBM}, \quad {\rm and } \ \ p^*=I+\frac1{\b} \ \ {\rm for \ ABM}.
$$
On the contrary, for mean-reverting processes (\ref{MRP}) and (\ref{SRMRP}) the function $\psi(x)$ is represented as infinite series, and  optimal threshold can be find only numerically.
\smallskip

So, Theorem 1 states that optimal threshold $p^*$ is a point of maximum for the function $h(x)=(x-I)/\psi(x)$. This implies the first-order optimality condition $h'(p^*)=0$, i.e. the equality (\ref{smooth}), and smooth-pasting principle:
$$
\ds V'_x(p^*;x)\bigr|_{x=p^*}=1.
$$

 In the next section we discuss smooth-pasting principle and appropriate free-boundary problem more closely.

\subsection{Threshold Strategies and Free-Boundary Problem}

It is almost common opinion (especially among engineers and economists) that solution to free-boundary problem always gives a solution to optimal stopping problem.

A free-boundary problem for the case of threshold strategies in investment timing problem  can be written as follows:
to find threshold $p^*\in (l,r)$ \ and twice differentiable function $H(x),\ l<x<p^*$, such that
\begin{eqnarray}
&&\L H(x)=\r H(x), \quad l<x<p^*;\label{fbp1} \\
&& H(p^*{-}0)= p^* -I, \quad H'(p^*{-}0)=1 \label{fbp2}.
\end{eqnarray}

If $\psi(x)$ is twice differentiable, then solution to the problem (\ref{fbp1})--(\ref{fbp2}) has the type
\beq\label{fbp0}
H(x)= \frac{p^*-I}{\psi(p^*)}\psi(x),\quad   l<x<p^*,
\eeq
where $\psi(x)$ is an increasing solution to ODE (\ref{diffur}) and $p^*$ satisfies the smooth-pasting condition (\ref{smooth}). In further, we will call such $p^*$ the solution to free-boundary problem.

According to Theorem 1 the optimal threshold in problem (\ref{optinv1}) must be a point of maximum of the function $h(x)=(x-I)/\psi(x)$, but smooth-pasting condition (\ref{smooth}) provides only a stationary point for $h(x)$. Thus, we can apply standard second-order optimality conditions to derive relations between solutions to investment timing problem and to  free-boundary problem.

Let $p^*$ be a solution to free-boundary problem (\ref{fbp1})--(\ref{fbp2}). If $p^*$ is also an optimal threshold in investment timing problem (\ref{optinv1}), then, of course, $h''(p^*)\le 0$. It means that
$$
\psi''(p^*)=-\frac{h''(p^*)\psi(p^*)+ 2h'(p^*)\psi'(p^*) }{h(p^*)}= -\frac{h''(p^*)\psi(p^*)}{h(p^*)} \ge 0.
$$
Thus, the inequality $\psi''(p^*)\ge 0$ may be viewed as necessary condition for a solution of  free-boundary problem to be optimal in investment timing problem.

The inverse relation between solutions can be state as follows.
\medskip

\noindent {\bf Statement 1.} {\it If \ $p^*$ is the unique solution to free-boundary problem \emph{(\ref{fbp1})--(\ref{fbp2})}, and $\psi''(p^*)> 0$, then  $p^*$ is optimal threshold in the problem \emph{(\ref{optinv1})} for all $x\in \I$.
}

\begin{proof} \ Since $h'(p^*)=0$ and $\psi''(p^*)> 0$ then $h''(p^*)=-h(p^*)\psi''(p^*)/\psi(p^*){<}0$. Therefore, $h'(p)$ strictly decreases at some neighborhood of $p^*$.

Then, it is easy to see that  $h'(p)>0$ for $p<p^*$ and $h'(p)<0$ for $p>p^*$, else  $h'(q)=0$ for some $q\neq p^*$, that contradicts to the uniqueness of solution to free-boundary problem (\ref{fbp1})--(\ref{fbp2}). So, conditions  (\ref{criteria0})--(\ref{criteria00}) hold and Theorem 1 gives the optimality of threshold ${p^*}$.
\end{proof}

The following result concerns the general case when free-boundary problem has several solutions.\medskip

\noindent{\bf Statement 2.} {\it Let  \ $p^*$ and \ $\tilde p$ \ are two solutions   to free-boundary problem \emph{(\ref{fbp1})--(\ref{fbp2})}, such that \ $\psi''(p^*)> 0$ and $(x-I)/\psi(x)\le (p^*-I)/\psi(p^*)$ for $l{<}x{<} p^*$. If $\tilde p >p^*$, and
 $\psi^{(k)}(\tilde p)=0$ \ $(k=2,...,n-1)$, $\psi^{(n)}(\tilde p)> 0$ for some $n>2$,
then ${p^*}$ is optimal threshold in the problem \emph{(\ref{optinv1})} for all $x\in \I$.
}
\medskip

\begin{proof} \
Let us prove that $h'(p)\le 0$ for all $p>p^*$. Inequality $\psi''(p^*)> 0$ implies (as above) that $h''(p^*)< 0$, and, therefore,  $h'(p)< 0$ for all $p^*<p<p_1$ with some $p_1$. If we suppose that $h'(p_2)>0$ for some $p_2 >p^*$, then there exists $p_0\in (p_1,p_2)$ such that $h'(p_0)=0$ and $h'(p)>0$ for all $p_0<p<p_2$.  Therefore, $p_0$ is another solution to free-boundary problem  (\ref{fbp1})--(\ref{fbp2}), and due to conditions of the Statement $h^{(k)}(p_0)=0$ \ $(k=2,...,n-1)$, $h^{(n)}(p_0)< 0$ for some $n>2$, that contradicts to positivity of $h'(p)$ for $p_0<p<p_2$.

Hence, $h'(p)\le 0$ for all $p>p^*$ and conditions  (\ref{criteria0})--(\ref{criteria00}) hold. Thus, according to Theorem 1, ${p^*}$ is optimal threshold in the problem (\ref{optinv1}).
\end{proof}

\medskip

\subsection{Optimal Strategies in Investment Timing Problem}

Now, return to `general' investment timing problem (\ref{optinv0}) over the class $\M$ of \emph{all stopping times}.

The sufficient conditions under which an optimal investment time in (\ref{optinv0}) will be a threshold stopping time were derived in \cite{Al}. In this section we give necessary and sufficient conditions (criterion) for optimality of threshold stopping time in investment timing problem (\ref{optinv0}).

To reduce some technical difficulties we assume below that drift $a(x)$ and diffusion $\s(x)$ of the underlying process $X_t$ are continuous functions.

\begin{theorem} \
Threshold stopping time  $\t_{p^*}$, $p^*{\in} (l,r)$, is optimal in the investment timing problem \emph{(\ref{optinv0})} for all $x{\in}\, \I$ if and only if the following conditions hold:
\begin{eqnarray}
&&
(p-I){\psi(p^*)}\le (p^*-I){\psi(p)}\quad \mbox{\rm  for } p<p^*;\label{criteria110}\\[3pt]
&& \psi(p^*)=(p^*-I)\psi'(p^*);
\label{criteria11}\\
&& a(p)\le \r (p- I) \quad \mbox{\rm  for } p>p^*, \label{criteria21}
\end{eqnarray}
where $\psi(x)$ is an increasing solution to ODE \emph{(\ref{diffur})} and $a(p)$ is the drift function of the process $X_t$.
\end{theorem}

\begin{proof}
Define the value function for the problem (\ref{optinv0}) over the class $\M$ of all stopping times
$$
V(x)=\sup_{\t\in\M}{\Ex}\left(X_\t - I\right) e^{-\r \t}\II_{\{\t{<}\infty\}}.
$$

i) Let conditions (\ref{criteria110})--(\ref{criteria21}) hold.

Take the function
$$
\Phi(x)=V({p^*};x)=\left\{
\begin{array}{ll}\ds\frac{p^*-I}{\psi(p^*)}
\psi(x), & \ {\rm for }\ x<p^*,\\[8pt]
x-I, &    \ {\rm for }\ x\ge p^*.
\end{array}
\right.
$$

Obviously, $\Phi(x)>0$ (due to condition (\ref{criteria11})) and $ V(x)\ge \Phi(x)$.

On the other hand, (\ref{criteria110}) implies
$$
\frac{p^*-I}{\psi(p^*)}\psi(x)\ge \frac{x-I}{\psi(x)}\psi(x)=x-I,
$$
therefore $\Phi(x)\ge x-I$ for all $x\in(l,r)$, i.e. $\Phi(x)$ is a majorant of payoff function $x-I$.

For any stopping time $\t\in \M$ and $N>0$ put $\tilde\t=\t{\wedge}N$. From It\^{o}--Tanaka--Meyer formula (see, e.g. \cite{KS}) we have:
\begin{eqnarray}
\Ex \Phi(X_{\tilde\t}) e^{-\r\tilde\t}&=& \Phi (x)+\Ex \int_0^{\tilde\t}(\L\Phi -\r\Phi)(X_t)e^{-\r t}dt\nonumber\\
&+&\frac12 \s^2(p^*)[\Phi'(p^*{+}0)-\Phi'(p^*{-}0)]\Ex \int_0^{\tilde\t} e^{-\r t}dL_t(p^*),
\end{eqnarray}
where $L_t(p^*)$ is the local time of the process $X_t$ at the point $p^*$.

By definition  we have
$$
\Phi'(p^*{+}0)-\Phi'(p^*{-}0)= 1-
\frac{p^*-I}{\psi(p^*)} \psi'(p^*)= 0
$$
due to (\ref{criteria11}).

Take $T_1=\{0\le t\le {\tilde\t}:\, X_t < p^*\}$, \ $T_2=\{0\le t\le {\tilde\t}:\, X_t > p^*\}$. We have:
\begin{eqnarray*}
&&\L\Phi(X_t) -\r\Phi(X_t) =\frac{p^*-I}{\psi(p^*)} \Bigl(\L\psi(X_t)-\r\psi(X_t)\Bigr)=0 \quad {\rm for } \ t\in T_1,\\
&&\L\Phi(X_t) -\r\Phi(X_t) =a(X_t) -\r (X_t-I)\le 0 \quad {\rm for } \  t\in T_2
\end{eqnarray*}
by definition of the function $\psi(x)$ and (\ref{criteria21})).

Then
\begin{eqnarray*}
\Ex \Phi(X_{\tilde\t}) e^{-\r\tilde\t}&\le &\Phi (x)+\Ex \left( \int\limits_{T_1}(\L\Phi {-}\r\Phi)(X_t)e^{-\r t}dt \right.
+ \left.\int\limits_{T_2}(\L\Phi {-}\r\Phi)(X_t)e^{-\r t}dt\right)\\
&\le &\Phi (x).
\end{eqnarray*}

Since $\Phi(X_{\tilde\t}) e^{-\r \tilde\t}\stackrel{\mbox{\scriptsize a.s.}}{\longrightarrow} \Phi(X_{\t}) e^{-\r\t}\II_{\{\t{<}\infty\}}$ when $N\to \infty$, then due to Fatou's Lemma : $\ds \Ex \Phi(X_{\t}) e^{-\r\t}\II_{\{\t{<}\infty\}}\le \Phi (x)$ for all $\t\in \M$ and $x\in \I$. Therefore, $\Phi(x)$ is $\r$-excessive function, which majorates payoff function $x-I$. Since, by Dynkin's characterization, value function $V(x)$ is the least $\r$-excessive majorant, then $ V(x)\le \Phi(x)$.

Therefore, $ V(x)= \Phi(x)=V({p^*};x)$, i.e. $\t_{p^*}$ is the optimal stopping time in problem (\ref{optinv0}) for all $x$.

ii)  Now, let $\t_{p^*}$ be optimal stopping time in the problem (\ref{optinv0}). Note, that $\t_{p^*}$ will be an optimal stopping time in the problem (\ref{optinv1}) also. Therefore, Theorem 1 implies  (\ref{criteria110}) and (\ref{criteria11}), since $p^*$ is point of maximum for the function $(x-I)/\psi(x)$.

Further, assume that inequality (\ref{criteria21}) is not true at some point $p_0 > p^*$, i.e. $a(p)> \r (p-I)$ in some interval $J\subset (p^*,r)$ (by virtue of continuity). For some $\tilde x\in J$ define $ \t=\inf\{t\ge 0:\, X_t\notin J\}$, where process $X_t$ starts from the point  $\tilde x$. Then for any $N>0$ from Dynkin's formula
$$
\E^{\tilde x} (X_{\t{\wedge}N}-I) e^{-\r (\t{\wedge}N)}= \tilde x-I+\E^{\tilde x} \int_0^{\t{\wedge}N}[a(X_t) -\r (X_t-I)]e^{-\r t}dt >\tilde x-I.
$$
Therefore, $V(\tilde x) > \tilde x-I$ that contradicts to $V(\tilde x)=V({p^*};\tilde x)=g(\tilde x)$,  since $\tilde x>p^*$.

\end{proof}

\begin{example}
Let $X_t$ be the process of geometric Brownian motion (\ref{GBM}). Then Theorem 2 implies that threshold stopping time  $\t_{p^*}$ will be optimal in the investment timing problem (\ref{optinv0}) over all investment times if and only if
$\ds
p^*=I{\b}/{(\b-1)},
$
where $\b$ is the positive root of the equation $\ds \textstyle\frac12\s^2\b(\b-1)+\a\b-\r = 0$.
\end{example}

\subsubsection*{Acknowledgments.} The work was supported by Russian Foundation for Basic Researches (project 15-06-03723) and Russian Foundation for Humanities (project 14-02-00036).


\begin{thebibliography}{4}

\bibitem{Al}  Alvarez, L.H.R.: Reward functionals, salvage values, and optimal stopping. Math. Methods Oper. Res. 54, 315--337 (2001)

\bibitem{A}  Arkin, V.I.: Threshold Strategies in Optimal Stopping Problem for One-Dimensional Diffusion Processes. Theory Probab. Appl. 59, 311--319 (2015)

\bibitem{AS}  Arkin, V.I., Slastnikov A.D.: Threshold stopping rules for diffusion processes and Stefan's problem. Dokl. Math. 86, 626-–629 (2012)

\bibitem{BS}  Borodin, A.N., Salminen, P.: Handbook of Brownian Motion -- Facts and Formulae. Birkhauser-Verlag (2002)

\bibitem{CM}  Crocce, F., Mordecki, E.: Explicit solutions in one-sided optimal stopping problems for one-dimensional diffusions. Stochastics, 86, 491-–509 (2014)

\bibitem{DP}  Dixit, A., Pindyck, R.S.: Investment under Uncertainty. Princeton University Press, Princeton (1994)

\bibitem{DPS}  Dixit, A., Pindyck, R.S., S{\o}dal, S.: A Markup Interpretation of Optimal Investment Rules. The Economic Journal. 109, 179--189 (1999)

\bibitem{JZ} Johnson, T.C., Zervos, M.: A discretionary stopping problem with applications to the optimal timing of investment decisions. Preprint, \url{https://vm171.newton.cam.ac.uk/files/preprints/ni05045.pdf} (2005)

\bibitem{IM} Ito, K., McKean, H.:
   Diffusion Processes and Their Sample Paths. Springer-Verlag, Berlin (1974)

\bibitem{KS} Karatzas, I., Shreve, S.~E.: Brownian Motion and Stochastic Calculus. Springer-Verlag, Berlin (1991)

\bibitem{MS}  McDonald, R.,  Siegel, D.: The value of waiting to invest. The Quarterly Journal of Economics, 101, 707--728 (1986)


\bibitem{PS} Peskir, G., Shiryaev, A.: Optimal Stopping and Free-Boundary Problems. Birkhauser (2006)



\end{thebibliography}
\end{document}